\documentclass[runningheads]{llncs}

\usepackage{amsmath}
\usepackage{amssymb}
\usepackage[funcfont=roman]{complexity}
\usepackage{enumerate}
\usepackage{stmaryrd}
\usepackage{macros}
\usepackage{graphicx}
\usepackage{wrapfig}       
\usepackage{marvosym}

\makeatletter
\RequirePackage[bookmarks,unicode,colorlinks=true]{hyperref}%
   \def\@citecolor{blue}%
   \def\@urlcolor{blue}%
   \def\@linkcolor{blue}%

\def\orcidID#1{\smash{\href{http://orcid.org/#1}{\protect\raisebox{-1.25pt}{\protect\includegraphics{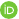}}}}}
\makeatother

\usepackage[capitalise]{cleveref}

\begin{document}


\title{On the Expressiveness of B\"uchi Arithmetic}


\author{Christoph Haase\inst{1}\thanks{Parts of this research were
    carried out while the first author was affiliated with the
    Department of Computer Science, University College London,
    UK.}\orcidID{0000-0002-5452-936X} (\Letter) \and Jakub
R{\'o}{\.{z}}ycki\inst{2}} \authorrunning{C. Haase and
    J. R{\'o}{\.{z}}ycki}

\institute{Department of Computer Science, University of Oxford, Oxford, UK
  \email{christoph.haase@cs.ox.ac.uk}
  \and Institute of Mathematics, University of Warsaw, Warsaw, Poland
}




\maketitle

\begin{abstract}
  We show that the existential fragment of B\"uchi arithmetic is
  strictly less expressive than full B\"uchi arithmetic of any base,
  and moreover establish that its $\Sigma_2$-fragment is already
  expressively complete. Furthermore, we show that regular languages
  of polynomial growth are definable in the existential fragment of
  B\"uchi arithmetic.
\end{abstract}


\section{Introduction}

This paper studies the expressive power of B\"uchi arithmetic, an
extension of Presburger arithmetic, the first-order theory of the
structure $\langle \N, 0, 1, + \rangle$. B\"uchi arithmetic
additionally allows for expressing restricted divisibility properties
while retaining decidability. Given an integer $p\ge 2$, \emph{B\"uchi
  arithmetic of base $p$} is the first-order theory of the structure
$\langle \N,0,1,+,V_p \rangle$, where $V_p$ is a binary predicate such
that $V_p(a,b)$ holds if and only if $a$ is the largest power of $p$
dividing $b$ without remainder, i.e., $a=p^k$, $a \mid b$ and $p \cdot
a \nmid b$.

Presburger arithmetic admits quantifier-elimination in the extended
structure $\langle \N,0,1,+,\{c|\cdot \}_{c>1}\rangle$ additionally
consisting of unary divisibility predicates $c|\cdot$ for every
$c>1$~\cite{Pre29}. It follows that the existential fragment of
Presburger arithmetic is expressively complete, since any predicate
$c|\cdot$ can be expressed using an additional existentially
quantified variable. We study the analogous question for B\"uchi
arithmetic and show, as the main result of this paper, that its
existential fragment is, in any base, strictly less expressive than
full B\"uchi arithmetic. Notably, this result implies that there does
not exist a quantifier-elimination result \emph{\`a la} Presburger for
B\"uchi arithmetic, i.e., any extension of B\"uchi arithmetic with
additional predicates definable in existential B\"uchi arithmetic does
not admit quantifier elimination.

A central result about B\"uchi arithmetic is that it is an automatic
structure: a set $M \subseteq \N^n$ is definable in B\"uchi arithmetic
of base $p$ if and only if $M$ is recognizable by a finite-state
automaton under a base $p$ encoding of the natural
numbers. Equivalently, $M$ is \emph{$p$-regular}. This result was
first stated by B\"uchi~\cite{B60}, albeit in an incorrect form, and
later correctly stated and proved by Bruy\`ere~\cite{B85}, see
also~\cite{BHMV94}. Villemaire showed that the $\Sigma_3$-fragment of
B\"uchi arithmetic is expressively complete~\cite[Cor.~2.4]{Vil92}. He
established this result by showing how to construct a
$\Sigma_3$-formula defining the language of a given finite-state
automaton. We observe that Villemaire's construction can actually be
improved to a $\Sigma_2$-formula and thus obtain a full
characterization of the expressive power of B\"uchi arithmetic in
terms of the number of quantifier alternations.

Our approach to separating the expressiveness of existential B\"uchi
arithmetic from full B\"uchi arithmetic in base $p$ is based on a
counting argument. Given a set $M\subseteq \N$, define the counting
function $d_M(n) := \# (M \cap \{ p^{n-1}, \ldots, p^{n}-1 \})$ which
counts the numbers of bit-length $n$ in base $p$ in $M$. If $M$ is
definable in existential B\"uchi arithmetic of base $p$, we show that
$d_M$ is either $O(n^c)$ for some $c\ge 0$, or at least $c \cdot p^n$
for some constant $c>0$ and infinitely many $n\in \N$. Since, for
instance, for $M_p\subseteq \N$ defined as the set of numbers with
$p$-ary expansion in the regular language $\{10, 01\}^*$, we have
$d_{M_p}(n)= \Theta(2^{n/2})$, and hence $M_p$ is not definable in
existential B\"uchi arithmetic of base $p$. However, $M_p$ being
$p$-regular implies that $M_p$ is definable by a $\Sigma_2$-formula of
B\"uchi arithmetic of base $p$.

We also show that existential B\"uchi arithmetic defines all regular
languages of polynomial density, encoded as sets of integers. Given a
language $L\subseteq \Sigma^*$, let the counting function $d_L\colon
\N \to \N$ be such that $d_L(n) := \#(L\cap \Sigma^n)$. Szilard et
al.~\cite{SYZS92} say that $L$ has \emph{polynomial density} whenever
$d_L(n)$ is $O(n^c)$ for some non-negative integer $c$. If moreover
$L$ is regular then Szilard et al. show that $L$ is represented as a
finite union of regular expressions of the form $v_0 w_1^* v_{1}
\cdots w_{k}^* v_k$ such that $0\le k \le c+1$, $v_0,w_1,v_1, \ldots,
v_k, w_k \in \Sigma^*$~\cite[Thm.~3]{SYZS92}. We show that existential
B\"uchi arithmetic defines any language represented by a regular
expression $v_0 w_1^* v_1 \cdots w_k^* v_k$, which implies that
existential B\"uchi arithmetic defines all regular languages of
polynomial density.

\section{Preliminaries}\label{sec:preliminaries}

Given $\vec v=(v_1,\ldots,v_d)\in \Z^d$, we denote by $\norm{\vec
  v}_\infty$ the maximum norm of $\vec v$, i.e., $\norm{\vec
  v}_\infty=\max\{ |v_1|, \ldots, |v_d| \}$. For a matrix $\mat A \in
\Z^{m\times d}$ with entries $a_{i,j}$, $1\le i\le m$, $1\le j \le d$,
we denote by $\norm{\mat A}_{1,\infty}$ the one-infinity norm of $\mat
A$, i.e., $\norm{\mat A}_{1,\infty}=\max \{ |a_{i,1}| + \cdots +
|a_{i,d}| : 1\le i \le m\}$.

Let $\Sigma$ be an alphabet and $w\in \Sigma^*$, we denote by $|w|$
the length of $w$. Given a set $U\subseteq \N$, we denote by $w^U :=
\{ w^u : u \in U \}$. Thus, for example, $w^*=w^{\N}$.

For an integer $p\ge 2$, let $\Sigma_p := \{ 0, \ldots, p-1\}$. We
view words over $\Sigma_p$ as numbers encoded in $p$-ary
most-significant bit first encoding. Tuples of numbers of dimension
$n$ can be encoded as words over the alphabet $\Sigma_p^n$. For
$w=\vec v_m\cdots \vec v_0 \in (\Sigma_p^n)^{m+1}$, we denote by
$\eval w_p \in \N^n$ the $n$-tuple
\[
\eval w_p := \sum_{i=0}^m \vec v_i \cdot p^i\,.
\]
We furthermore define $\eval{\varepsilon}_p := 0$. Note that $\eval \cdot_p$
is not injective since, e.g., $01$ and $001$ both encode the number
one. Given $L \subseteq (\Sigma_p^n)^*$, we define
\[
\eval L_p := \left \{ \eval w_p : w \in L \right \} \subseteq \N^n\,.
\]

\paragraph{Automata.}
A \emph{deterministic automaton} is a tuple
$A=(Q,\Sigma,\delta,q_0,F)$, where
\begin{itemize}
  \item $Q$ is a set of \emph{states},
  \item $\Sigma$ is a finite alphabet,
  \item $\delta\colon Q \times \Sigma \to Q \cup \{ \bot \}$, where
    $\bot\not\in Q$, is the \emph{transition function},
  \item $q_0\in Q$ is the \emph{initial state}, and
  \item $F \subseteq Q$ is the set of \emph{final states}.
\end{itemize}
For states $q,r\in Q$ and $u\in \Sigma$, we write $q \xrightarrow{u}
r$ if $\delta(q,u)=r$, and extend $\xrightarrow{}$ inductively to
words by stipulating, for $w\in \Sigma^*$ and $u\in \Sigma$, that
$q\xrightarrow{w \cdot u} r$ if there is $s\in Q$ such that
$q\xrightarrow{w} s \xrightarrow{u} r$. The \emph{language of $A$} is
defined as $L(A) = \{ w\in \Sigma^* : q_0 \xrightarrow{w} q_f, q_f\in
F \}$.


Note that \emph{a priori} we allow automata to have infinitely many
states and to have partially defined transition functions (due to the
presence of $\bot$ in the co-domain of $\delta$). If $Q$ is finite
then we call $A$ a \emph{deterministic finite automaton (DFA)}, and if
in addition $\Sigma=\Sigma_p^n$ for some $p\ge 2$ and $n\ge 1$ then
$A$ is called a \emph{$p$-automaton}. Throughout this paper, we
assume, without loss of generality, that all states of a DFA are live,
i.e., every state is reachable from the initial state and can reach an
accepting state.

\paragraph{Arithmetic theories.}

As stated in the introduction, Presburger arithmetic is the
first-order theory of the structure $\langle \N,0,1,+ \rangle$, and
B\"uchi arithmetic of base $p$ the first-order theory of the extended
structure $\langle \N,0,1, +, V_p \rangle$. We write atomic formulas
of Presburger arithmetic as $\vec a \cdot \vec x = c$, where $\vec a =
(a_1,\ldots,a_d)^\intercal$ with $a_i\in \Z$, $c\in \Z$, and $\vec x =
(x_1,\ldots,x_d)$ is a vector of unknowns. In B\"uchi arithmetic we
additionally have atomic formulas $V_p(x,y)$ for the unknowns $x$ and
$y$. For technical convenience, we assert that $V_p(x,0)$ never
holds.\footnote{Other conventions are possible, e.g., asserting that
  $V_p(x,0)$ holds if and only if $x=1$ as in~\cite{BHMV94}, but this
  does not change the sets of numbers definable in B\"uchi
  arithmetic.} We write $\Phi(x)$ or $\Phi(\vec x)$ to indicate that
$x$ or a vector of unknowns $\vec x$ occurs free in $\Phi$. If there
are further free variables in $\Phi$, we assume them to be implicitly
existentially quantified.

We may without loss of generality assume that no negation symbol
occurs in a formula of B\"uchi arithmetic. First, we have $\neg(\vec a
\cdot \vec x= c) \equiv \vec a\cdot \vec x \le c-1 \vee \vec a \cdot
\vec x \ge c+1$, and the order relation $\le$ can easily be expressed
by introducing an additionally existentially quantified
variable. Moreover, we have
\[
\neg V_p(x,y) \equiv y = 0 \vee \exists z\colon V_p(z,y) \wedge \neg(x=z)\,.
\]
Finally, $P_p(x) := V_p(x,x)$ denotes the macro asserting that $x$ is
a power of $p$.

Given a formula $\Phi(\vec x)$ of B\"uchi arithmetic of base $p$, we
define
\[
\eval{\Phi(\vec x)}_p := \left\{ \vec m\in \N^d : \Phi[\vec m/\vec x]
\text{ is valid} \right \},
\]
where, for $\vec m=(m_1,\ldots,m_d)$ and $\vec x = (x_1,\ldots,x_d)$,
$\Phi[\vec m/\vec x]$ is the formula obtained from replacing every
$x_i$ by $m_i$ in $\Phi$. The set of sets of numbers definable in
Presburger arithmetic is denoted by
\[
\mathbf{PA}:=\{ \eval{\Phi(x)} : \Phi(x) \text{ is a formula of
  Presburger arithmetic} \}\,.
\]
Analogously, we define the sets of numbers definable in fragments of
B\"uchi arithmetic of base $p$ with a fixed number of
quantifier-alternations as
\[
\Sigma_i\text{-}\mathbf{BA}_p := \left \{ \eval{\Phi(x)}_p : \Phi(x)
\text{ is a } \Sigma_i\text{-formula of B\"uchi arithmetic of base } p
\right\}\,.
\]
Finally, $\mathbf{BA}_p := \bigcup_{i\ge 1} \Sigma_i$-$\mathbf{BA}_p$
denotes the sets of numbers definable in B\"uchi arithmetic of base
$p$.

For separating existential B\"uchi arithmetic from full B\"uchi
arithmetic, we employ some tools from enumerative combinatorics.  As
defined in~~\cite{Woods14}, a formula of \emph{parametric Presburger
  arithmetic} with parameter $t$ is a formula of Presburger arithmetic
$\Phi_t$ in which atomic formulas are of the form $\vec a \cdot \vec x
= c(t)$, where $c(t)$ is a univariate polynomial with indeterminate
$t$ and coefficients in $\Z$. For $n\in \N$, we denote by $\Phi_n$ the
formula of Presburger arithmetic obtained from replacing $c(t)$ in
every atomic formula of $\Phi_t$ by the value of $c(n)$. We associate
to a formula $\Phi_t(\vec x)$ the counting function $\#\Phi_t(\vec x)
\colon \N \to \N \cup \{ \infty \}$ such that
\[
\#\Phi_t(\vec x)(n) := \#\eval{\Phi_n(\vec x)}.
\]
Throughout this paper, we constraint ourselves to formulas
$\Phi_t(\vec x)$ of parametric Presburger arithmetic in which $c(t)$
is the identity function and $\#\Phi_t(\vec x)(n)$ is finite for all
$n \in \N$.
\begin{definition}
  A function $f\colon \N \to \Q$ is an \emph{eventual
    quasi-polynomial} if there exist a threshold $t\in \N$ and
  polynomials $p_0,\ldots,p_{m-1}\in \Q[x]$ such that for all $n>t$,
  $f(n) = p_i(n)$ whenever $n \equiv i \bmod m$.
\end{definition}
Given an eventual quasi-polynomial $f$ with threshold $t$ and $n>t$,
we denote by $f_n$ the polynomial $p_i$ such that $n \equiv i \bmod
m$. We say that the polynomials $p_0,\ldots,p_{m-1}$ \emph{constitute}
the eventual quasi-polynomial $f$.
A result by Woods~\cite[Thm.\ 3.5(b)]{Woods14} shows that the counting
functions associated to parametric Presburger formulas as defined
above are eventual quasi-polynomial.
\begin{proposition}[Woods]\label{prop:pa-eqp}
  Let $\Phi_t(\vec x)$ be a formula of parametric Presburger
  arithmetic. Then $\#\Phi_t(\vec x)$ is an eventual quasi-polynomial.
\end{proposition}

\paragraph{Semi-linear sets.}

A result by Ginsburg and Spanier establishes that the sets of numbers
definable in Presburger arithmetic are semi-linear sets~\cite{GS64}. A
\emph{linear set} in dimension $d$ is given by a base vector $\vec b
\in \N^d$ and a finite set of period vectors $P=\{ \vec p_1,\ldots,
\vec p_n \}\subseteq \N^d$ and defines the set
\[
L(\vec b,P) := \left\{ \vec b + \lambda_1 \cdot \vec p_1 + \cdots +
\lambda_n \cdot \vec p_n : \lambda_i \in \N, 1\le i\le n \right \}.
\]
A \emph{semi-linear set} is a finite union of linear sets. For a
finite $B \subseteq \N^d$, we write $L(B,P)$ for $\bigcup_{\vec b \in
  B} L(\vec b, P)$. Semi-linear sets of the form $L(B,P)$ are called
hybrid linear sets in~\cite{CH16}, and it is known that the set of
non-negative integer solutions of a system of linear Diophantine
inequalities $S\colon \mat A \cdot \vec x \ge \vec c$ is a hybrid
linear set~\cite{CH16}.

Semi-linear sets in dimension one are also known as \emph{ultimately
  periodic sets}. In this paper, we represent an ultimately periodic
set as a four-tuple $U=(t,\ell,B,R)$, where $t\ge 0$ is a
\emph{threshold}, $\ell>0$ is a \emph{period}, $B\subseteq \{0,\ldots,
t-1\}$ and $R\subseteq \{0,\ldots,\ell-1\}$, and $U$ defines the set
\[
\eval{U} := B \cup \{ t + r + \ell \cdot i : r\in R, i\ge 0 \}\,.
\]

\section{The inexpressiveness of  existential B\"uchi arithmetic}

We now establish the main result of this paper and show that the
existential fragment of B\"uchi arithmetic is strictly less expressive
than general B\"uchi arithmetic.
\begin{theorem}\label{thm:main}
  For any base $p\ge 2$, $\Sigma_1$-$\mathbf{BA}_p \neq \mathbf{BA}_p$. In
  particular, there exists a fixed regular language $L \subseteq
  \{0,1\}^*$ such that $\eval L_p\in \mathbf{BA}_p \setminus
  \Sigma_1$-$\mathbf{BA}_p$ for every base $p\ge 2$.
\end{theorem}
Given a set $M\subseteq \N$, recall that for a fixed base $p\ge 2$,
$d_M(n)$ counts the numbers of bit-length $n$ in base $p$ in $M$. As
already discussed in the introduction, we prove \cref{thm:main} by
characterizing the growth of $d_M$ for sets $M$ definable in B\"uchi
arithmetic.

For any formula $\Phi(x)$ of existential B\"uchi arithmetic in prenex
normal form, we can with no loss of generality assume that its matrix
is in disjunctive normal form, i.e., a disjunction of \emph{systems of
  linear Diophantine equations with valuation constraints}, each of
the form
\[
\mat A \cdot \vec x = \vec c \wedge \bigwedge_{i\in I} V_p(x_i,y_i), 
\]
where the $x_i$ and $y_i$ are unknowns from the vector of unknowns
$\vec x$. For $M=\eval{\Phi(x)}_p$, in order to determine the growth
of $d_M$, it suffices to determine the maximum growth occurring in any
of its systems of linear Diophantine equations with valuation
constraints in the matrix of $\Phi(x)$, which in turn can be obtained
by analyzing the growth of the number of words accepted by a
$p$-automaton defining the set of solutions of such a system.

Let $S\colon \mat A\cdot \vec x = \vec c$ be a system of linear
Diophantine equations such that, throughout this section, $\mat A$ is
an $m\times d$ integer matrix, and fix a base $p \ge 2$. Following
Wolper and Boigelot~\cite{WB00}, we define an automaton
$A:=(Q,\Sigma_p^d,\delta,\vec q_0,F)$ whose language encodes all
solutions of $S$ over the alphabet $\Sigma_p$:
\begin{itemize}    
\item $Q := \Z^m$,
\item $\delta(\vec q,\vec u) := p \cdot \vec q + \mat A \cdot \vec u$
  for all $\vec q\in Q$ and $\vec u \in \Sigma_p^d$,
\item $\vec q_0 := \vec 0$, and
\item $F := \{ \vec c \}$.
\end{itemize}
As discussed in~\cite{WB00}, see also~\cite{GHW19}, only states $\vec
q$ such that $\norm{\vec q}_\infty \le \norm{\mat A}_{1,\infty}$ and
$\norm{\vec q}_\infty \le \norm{\vec c}_\infty$ can reach the
accepting state. Hence, all words $w \in (\Sigma_p^d)^*$ such that
${\mat A \cdot \eval w} = \vec c$ only visit a finite number of states
of $A$, and to obtain the $p$-automaton $A(S)$ defining the sets of
solutions of $S$ we subsequently restrict $Q$ to only such states. The
following lemma recalls an algebraic characterization of the
reachability relation of $A(S)$ established in the proof of
Proposition~14 in \cite{GHW19}.
\begin{lemma}\label{lem:reach-char}
  Let $\vec q,\vec r \in \Z^m$ be states of $A(S)$, $w \in
  (\Sigma_{p}^d)^n$ and $\vec x = \eval w_p$. Then $\vec q
  \xrightarrow{w} \vec r$ if and only if there is $y\in \N$ such that
  \[\vec q =  \vec r \cdot y + \mat A
  \cdot \vec x,~ \norm{\vec x}_\infty < y,~ y = p^n.\]
\end{lemma}

Let $x$ be a distinguished variable of $\vec x$. For a word $w \in
(\Sigma_p^d)^*$ encoding solutions of $S$, denote by $\pi_x(w)$ the
word $v \in \Sigma_p^*$ obtained from projecting $w$ onto the
component of $w$ corresponding to $x$. Let $q$ be a state of a
$p$-automaton $A$, define the counting function $C_{q,x} \colon \N
\to \N$ as
\[
C_{q,x}(n) := \# \left\{ \pi_x(w) : q \xrightarrow{w} q, w\in
(\Sigma_p^d)^n \right\}.
\]
We now show that for $p$-automata arising from systems of linear
Diophantine equations, $C_{q,x}$ can be obtained from an eventual
quasi-polynomial.
\begin{lemma}\label{lem:eq-eqp}
  For the $p$-automaton $A(S)$ associated to $S\colon \mat A\cdot \vec
  x = \vec c$ with states $Q$ and all $q\in Q$, there is an eventual
  quasi-polynomial $f$ such that $C_{q,x}(n) = f(p^n)$ for all $n \in
  \N$. Moreover, for all sufficiently large $n\in \N$, $f_{p^n}$ is a
  linear polynomial.
\end{lemma}
\begin{proof}
  Let $q=\vec q\in \Z^d$. By \cref{lem:reach-char}, $\vec q
  \xrightarrow{w} \vec q$ for $w\in (\Sigma_p^d)^n$ if and only if
  there is a $y\in \N$ such that
  \[\vec q = \vec q\cdot y + \mat A
  \cdot \vec x,~ \norm{\vec x}_\infty < y,~ y = p^n,\]
  where $\vec x = \eval{w}_p$. The set of solutions of $S'\colon \mat
  A\cdot \vec x + \vec q\cdot y = \vec q, \norm{\vec x}_\infty < y$
  is a hybrid linear set $L(D,R)\subseteq \N^{d+1}$. Let $L(B,P)
  \subseteq \N^2$ be obtained from $L(D,R)$ by projecting onto the
  components corresponding to $x$ and $y$, and assume that $x$
  corresponds to the first and $y$ to the second component of
  $L(B,P)$. Let $M_t := \N \times \{ t \}$ and \[f(t) := \# (L(B,P)
  \cap M_t)\,.\] Observe that $C_{q,x}(n) = f(p^n)$ and that $f(n)$ is
  finite for all $n\in \N$ due to the constraint $x<y$. Let $P=\{ \vec
  p_1,\ldots, \vec p_k \}$, the following formula of parametric
  Presburger arithmetic defines $L(B,P)\cap M_t$:
  \[
  \Phi_t(x,y) := \exists z_1\cdots \exists z_k\colon \bigvee_{\vec b
    \in B} \begin{pmatrix}x\\y\end{pmatrix} =
    \vec b + \sum_{i=1}^k  \vec p_i \cdot z_i \wedge y = t
  \]
  Thus, $f=\#\Phi_t(x,y)$ and, by application of \cref{prop:pa-eqp},
  $f$ is an eventual quasi-polynomial.

  Since $C_{q,x}(n) \le p^n -1$ for all $n\in \N$, we in particular
  have that all polynomials $f_{p^n}$ constituting $f$ are linear as
  they would otherwise outgrow $C_{q,x}$.\qed
\end{proof}

The next step is to lift \cref{lem:eq-eqp} to systems of linear
Diophantine equations with valuation constraints. To this end, we
define a DFA whose language encodes the set of all solutions of
predicates of the form $V_p(x,y)$. Formally, for $S\colon V_p(x,y)$ we
define $A(S) := (Q, \Sigma_p^d, \delta, q_0, F)$ such that
\begin{itemize}
\item $Q := \{ 0, 1\}$,
\item $\delta(0, \vec u) := 0$ for all $\vec u\in \Sigma_p^d$ such that
  $\pi_x(\vec u)=0$,
\item $\delta(0, \vec u) := 1$ for all $\vec u\in \Sigma_p^d$ such
  that $\pi_x(\vec u)=1$ and $\pi_y(\vec u)>0$,
\item $\delta(1, \vec u) := 1$ for all $\vec u\in \Sigma_p^d$ such
  that $\pi_x(\vec u)=\pi_y(\vec u)=0$,
\item $q_0 := 0$, and
\item $F:= \{ 1 \}$.
\end{itemize}
For $S\colon \mat A \cdot \vec x = \vec c \wedge \bigwedge_{1 \le i
  \le \ell} V_p(x_i,y_i)$, we denote by $A(S)$ the DFA that can be
obtained from the standard product construction on all DFA for the
atomic formulas of $S$. Hence, the set of states of $A(S)$ is a finite
subset of $\Z^m \times \{0,1\}^\ell$. We now show that the number of
words along a cycle of $A(S)$ can also be obtained from an eventual
quasi-polynomial.
\begin{lemma}\label{lem:eqv-eqp}
  Let $S$ be a system of linear Diophantine equations with valuation
  constraints with the associated DFA $A(S)$ with states $Q$, and let
  $q \in Q$. There is an eventual quasi-polynomial $f$ such that
  $C_{q,x}(n)=f(p^n)$. Moreover, $f_{p^n}$ is a linear polynomial for
  all $n \in \N$.
\end{lemma}
\begin{proof}
  Let $S\colon \mat A \cdot \vec x = \vec c \wedge \bigwedge_{1 \le i
    \le \ell} V_p(x_i,y_i)$, we have $Q\subseteq \Z^m\times
  \{0,1\}^{\ell}$ and thus $q=(\vec q, b_1, \ldots, b_\ell) \in
  Q$. Any self-loop $q\xrightarrow{w}_S q$ with $q=(\vec q, b_1,
  \ldots, b_\ell)$ is a self-loop for the DFA induced by the system of
  linear Diophantine equations $\mat A \cdot \vec x = \vec c$ with the
  additional requirement that $\pi_{x_i}(\eval w_p)=0$ for all $1\le
  i\le \ell$ and furthermore $\pi_{y_i}(\eval w_p)=0$ whenever $b_i =
  1$. Thus $(\vec q, \vec 0) \xrightarrow{w}_{S'} (\vec q, \vec 0)$
  where
  \[
  S'\colon \mat A\cdot \vec x = \vec c \wedge \bigwedge_{1\le
    i\le \ell} x_i = 0 \wedge \bigwedge_{1\le i\le \ell, b_i=1}
  y_i=0\,.
  \]
  Conversely, $(\vec q, \vec 0) \xrightarrow{w}_{S'} (\vec q, \vec 0)$
  immediately gives $q \xrightarrow{w}_S q$. The statement is now an
  immediate consequence of the application of \cref{lem:eq-eqp} to
  $S'$.\qed
\end{proof}
We will from now on implicitly apply~\cref{lem:eqv-eqp}. As a first
application, we show that \cref{lem:eqv-eqp} allows us to classify the
DFA associated to a system of linear Diophantine equations with
valuation constraints.
\begin{lemma}\label{lem:dfa-properties}
  The DFA $A(S)$ associated to a system of linear Diophantine
  equations with valuation constraints $S$ with states $Q$ has either
  of the following properties:
  \begin{enumerate}[(i)]
  \item there is $q\in Q$ such that $C_{q,x}$ is an eventual
    quasi-polynomial $f$ and $f_{p^n}$ is a non-constant polynomial
    for infinitely many $n\in \N$; or
  \item there is a constant $d\ge 0$ such that $C_{q,x}(n) \le d$ for all
    $q\in Q$ and $n \in \N$.
  \end{enumerate}
\end{lemma}
\begin{proof}
  Suppose $A(S)$ has Property~(i). For a contradiction, suppose $d\ge
  0$ exists. Let $f$ be the eventual quasi-polynomial from
  Property~(i). Every non-constant polynomial $f_{p^n}$ constituting
  $f$ is of the form $a \cdot x + b$ with $a>0$. As there are
  infinitely many such $n$, there is some linear polynomial $g(x)=a
  \cdot x + b$ such that $g=f_{p^n}$ for infinitely many $n\in \N$.
  Hence $g(p^n) > d$ for some sufficiently large $n\in \N$.
  
  For the converse, suppose that $A(S)$ does not have
  Property~(i). Then there are $\ell, m > 0$ such that all $f_{p^n}$
  are constant polynomials bounded by some value $m\in \N$ for all
  $n\ge \ell$, $q \in Q$ and $f=C_{q,x}$. Hence we can choose $d =
  \max(\left\{ C_{q,x}(n) : q\in Q, 0 < n \le \ell \right \} \cup \{ m
  \}).$\qed
\end{proof}

We are now in a position to prove a dichotomy of the growth of the
number of words accepted by a DFA corresponding to a system of linear
Diophantine equations with valuation constraints.

\begin{lemma}\label{lem:eqv-growth}
  Let $S$ be a fixed system of linear Diophantine equations with
  valuation constraints with the associated DFA $A(S)$. Let
  $L=\pi_x(L(A(S)))$, then either
  \begin{enumerate}[(i)]
  \item $d_L(n) \ge c \cdot p^n$ for some fixed constant $c>0$ and
    infinitely many $n\in \N$; or
  \item $d_L(n)=O(n^c)$ for some fixed constant $c \ge 0$.
  \end{enumerate}
\end{lemma}
\begin{proof}
  Let $A(S)$ have the set of states $Q$, initial state $q_0$ and final
  state $q_f$. The DFA $A(S)$ has one of the two properties stated in
  \cref{lem:dfa-properties}.

  If $A(S)$ has the Property~(i) of \cref{lem:dfa-properties} then
  consider $q\in Q$ such that $C_{q,x}$ is an eventual
  quasi-polynomial $f$ such that $f_{p^n}$ is non-constant for
  infinitely many $n\in \N$, and let $i_1 < i_2< \ldots \in \N$ be
  such that all $f_{p^{i_j}}$ are the same non-constant polynomial $a
  \cdot x + b$. Consider $v$ and $w$ such that $q_0 \xrightarrow{v} q
  \xrightarrow{w} q_f$. Then for all sufficiently large $j$ we have
  \[
  d_L(i_j+|v|+|w|) \ge a \cdot p^{i_j} + b \ge c \cdot p^{(i_j + |v| +
    |w|)}
  \]
  for some fixed constant $c>0$.

  Otherwise, $A(S)$ has the Property~(ii) of
  \cref{lem:dfa-properties}, and there is some fixed $d\ge 0$ such
  that $C_{q,x}(n) \le d$ for all $n\in \N$ and $q\in Q$. Every $w\in
  L$ such that $|w|=n$ can uniquely be decomposed as
  $w=v_0w_1v_1w_2\cdots w_kv_k$ for some $k\le |Q|$ such that
  \begin{equation}\label{eqn:decomposition}
    q_0 \xrightarrow{v_0} q_{a_1} \xrightarrow{w_1} q_{a_1} \xrightarrow{v_1}
    q_{a_2} \xrightarrow{w_2} q_{a_2} \xrightarrow{v_2} q_{a_3} \cdots
    \xrightarrow{w_k} q_{a_k} \xrightarrow{v_k} q_{a_{k+1}},
  \end{equation}
  where $q_{a_{k+1}} = q_f$, $q_{a_i}\neq q_{a_j}$ for all $i \neq j$
  and each $q_{a_i} \xrightarrow{v_i} q_{a_{i+1}}$ corresponds to a
  loop-free path in $A(S)$. Since $C_{q,x} \le d$, there are at most
  $d^k \le d^{(\#Q)}$ words $u \in L$ of length $n$ that have the same
  sequence of states in the decomposition of \cref{eqn:decomposition}
  at the same position where they occur in $w$. Moreover, there are at
  most ${n \choose 2k} \le {n \choose 2\cdot \#Q} \le n^{(2\cdot
    \#Q)}$ possibilities at which the states $q_{a_i}$ can appear in
  any $u \in L$ of length $n$ for any particular sequence of states in
  the decomposition of \cref{eqn:decomposition}. Finally, there are at
  most $(\#Q)^{(\#Q)}$ such sequences. We thus derive
  \[
  d_L(n) \le (\#Q)^{\#Q} \cdot n^{(2\cdot\#Q)} \cdot d^{(\#Q)} = O(n^c)
  \]
  for some constant $c\ge 0$.\qed
\end{proof}

\begin{corollary}\label{cor:eba-density}
  Let $\Phi(x)$ be a fixed formula of existential B\"uchi arithmetic
  of base $p\ge 2$. Let $M=\eval{\Phi(x)}_p$, then either:
  \begin{enumerate}[(i)]
  \item $d_M(n) \ge c \cdot p^n$ for some fixed constant $c>0$ and
    infinitely many $n \in \N$; or
  \item $d_M(n) = O(n^c)$ for some fixed constant $c \ge 0$.
  \end{enumerate}
\end{corollary}
\begin{proof}
  Without loss of generality we may assume that $\Phi(x)$ is in
  disjunctive normal form such that $\Phi(x)=\bigvee_{i\in I}
  \Phi_i(x)$ and each $\Phi_i(x)$ is a system of linear Diophantine
  equations with valuation constraints $S_i$. For
  $M_i=\eval{\Phi_i(x)}_p$, we obtain $d_{M_i}$ by application of
  \cref{lem:eqv-growth}. If there is a constant $c\ge 0$ such that
  $d_{M_i}=O(n^c)$ for all $i\in I$ then $d_M = O(n^c)$. Otherwise, if
  there is some $i\in I$ such that $d_{M_i}(n) \ge c \cdot p^n$ for
  some constant $c>0$ and infinitely many $n\in \N$ then $d_M(n) \ge
  c\cdot p^n$ for infinitely many $n\in \N$.\qed
\end{proof}

As an immediate consequence of \cref{cor:eba-density}, we obtain:
\begin{corollary}\label{cor:eba-definability}
  Let $p\ge 2$ and $M \subseteq \N$ such that $f=o(d_M)$ for any
  $f=O(n^c)$, $c\ge 0$, and $d_M = o(p^n)$. Then $M \not\in
  \Sigma_1\text{-}\mathbf{BA}_p$.
\end{corollary}
For any $p\ge 2$, consider $L=\{01, 10\}^* \subseteq \Sigma_p^*$ and
$M=\eval L_p$. We have $d_M(n)=\Theta(2^{n/2})$, and thus
\cref{cor:eba-definability} yields $M \not \in
\Sigma_1$-$\mathbf{BA}_p$. However, since $M$ is $p$-regular, we have $M
\in \mathbf{BA}_p$. This concludes the proof of \cref{thm:main}.

\section{Expressive completeness of the $\Sigma_2$-fragment of B\"uchi
  arithmetic}

For a regular language $L \subseteq (\Sigma_p^d)^*$ given by a DFA,
Villemaire shows in the proof of Theorem~2.2 in~\cite{Vil92} how to
construct a $\Sigma_3$-formula of B\"uchi arithmetic $\Phi_L(\vec x)$
such that $\eval{\Phi_L(\vec x)}_p=\eval{L}_p$. This construction is
modularized and relies on an existential formula $\Phi_{p,j}(x,y)$
expressing that \emph{``$x$ is a power of $p$ and the coefficient of
  this power of $p$ in the representation of $y$ in base $p$ is
  $j$''}:
\begin{multline*}
  \Phi_{p,j}(x,y) \equiv P_p(x) \land \exists t\, \exists u\,
  \exists z\colon \big( y =
  z + j\cdot x + t) \land (z < x) \land\\
  \wedge ( (V_p(u,t) \wedge x < u) \lor t = 0)\,.
\end{multline*}
The only reason why $\Phi_L(\vec x)$ in~\cite{Vil92} is a
$\Sigma_3$-formula is that $\Phi_{p,j}(x,y)$ appears in an implication
both as antecedent and as consequent inside an existential
formula. Thus, if one could additionally define $\Phi_{p,j}(x,y)$ by a
$\Pi_1$-formula then $\Phi_L(\vec x)$ immediately becomes a
$\Sigma_2$-formula. That is, however, not difficult to achieve by
defining:
\begin{multline*}
  \widetilde{\Phi}_{p,j}(x,y) := P_p(x) \land \forall s\, \forall
  t\, \forall u\, \forall z\colon \\ \Big(\neg(s = z + j\cdot x + t) \lor
  (z \geq x) \lor ( \neg V_p(u,t) \vee x \geq u) \land \neg(t =
  0))\Big) \rightarrow \neg(s = y)\,.
\end{multline*}
Note that the order relation can also be expressed by a universal
formula: $x \le y$ if and only if $\forall z\colon (y+z=x) \rightarrow (z
= 0)$. Thus, $\widetilde{\Phi}_{p,j}(x,y)$ is indeed a $\Pi_1$
formula.

Combining $\widetilde{\Phi}_{p,j}(x,y)$ with the results
in~\cite{Vil92}, we obtain that the $\Sigma_2$-fragment of B\"uchi
arithmetic is expressively complete.
\begin{theorem}
  For any base $p\ge 2$, $\Sigma_2$-$\mathbf{BA}_p=\mathbf{BA}_p$.
\end{theorem}

\section{Existential B\"uchi arithmetic defines regular languages of polynomial growth}
For a language $L\subseteq \Sigma^*$, Szilard et al.~\cite{SYZS92} say
that $L$ has \emph{polynomial growth} if $d_L(n)=O(n^c)$ for some
constant $c \ge 0$ and all $n\in \N$. One of the main results
of~\cite{SYZS92} is that a regular language $L$ has polynomial growth
if and only if $L$ can be represented as a finite union of regular
expressions of the form
\begin{equation}\label{eqn:poly-growth-re}
v_0 w_1^* v_{1} \cdots v_{k-1} w_{k}^* v_k\,.
\end{equation}
Denote by \[\mathbf{PREG}_p :=
\big\{ \eval L_p : L \subseteq \Sigma_p^*,~ L\text{ is a regular
  language of polynomial growth}\big\}\] the numerical encoding of all
regular languages of polynomial growth in base $p$.
We show in this section that existential B\"uchi arithmetic defines
any regular language of the form in~\cref{eqn:poly-growth-re}. This
immediately gives the following theorem.
\begin{theorem}\label{thm:preg-ba}
  For any base $p\ge 2$, $\mathbf{PREG}_p \subseteq
  \Sigma_1$-$\mathbf{BA}_p$.
\end{theorem}
We first require a couple of abbreviations. Define
\[
W_p(x,y) := P_p(y) \wedge x<y \le p\cdot x,
\]
which expresses that $y$ is the smallest power of $p$ strictly greater
than $x$.

Let $\ell > 0$, Lohrey and Zetzsche introduce in~\cite{LZ20} the
predicate $S_\ell(x,y)$ which holds whenever
\[
x = p^r \text{ and } y = p^{r + \ell\cdot i} \text{ for some } i,r \ge
0\,.
\]
They show that $S_\ell(x,y)$ is definable in existential B\"uchi
arithmetic. Since $y = p^{\ell \cdot i} \cdot x$ if and only if $y
\equiv x \bmod (p^\ell-1)$, one can obtain $S_\ell$ as
\[
S_\ell(x,y) := P_p(x) \wedge P_p(y) \wedge \exists z\colon (y-x =
(p^\ell - 1) \cdot z) \wedge y\ge x\,.
\]
We slightly generalize $S_\ell$. Let $U \subseteq \N$, define the
predicate $S_U(x,y)$ to hold whenever
\[
x = p^r \text{ and } y = p^{r+u} \text{ for some } r\ge 0 \text{ and }
u\in U\,.
\]
\begin{lemma}
  For any ultimately periodic set $U\subseteq \N$, the predicate
  $S_U(x,y)$ is definable in existential B\"uchi arithmetic
\end{lemma}
\begin{proof}
  Suppose that $U$ is given as $(t,\ell,B,R)$, we define
  \[
  S_U(x,y) := P_p(x) \wedge P_p(y) \wedge \bigvee_{b\in B} y =
  p^b \cdot x \vee \bigvee_{r\in R} S_{\ell}(p^{t+r} \cdot x, y)\,.
  \]\qed
\end{proof}

Towards proving \cref{thm:preg-ba}, we now show that we can define
$\eval{w^*}_p$ for any $w\in\Sigma_p$.
\begin{lemma}\label{lem:w-star}
  For any $w\in \Sigma_p^*$, $\eval{w^*}_p$ is definable by a formula
  of existential B\"uchi arithmetic $\Phi_{w^*}(x)$.
\end{lemma}
\begin{proof}
  Let $m = p^\ell$ be the smallest power of $p$ greater than $\eval
  w_p$. Then for any $k>0$,
  \[
  \eval{w^k}_p=\eval{w}_p \cdot \sum_{i=0}^{k-1} m^i = \eval{w}_p
  \cdot \frac{m^k-1}{m-1}\,.
  \]
  It follows that $\eval{w^*}_p$ is defined by \[\Phi_{w^*}(x) :=
  x = 0 \vee
  \exists y\colon S_\ell(m,y) \wedge (m-1) \cdot x = \eval w_p \cdot
  (y-1)\,.\]\qed
\end{proof}
Building upon \cref{lem:w-star}, we now show that, for any $w \in
\Sigma_p$, we can define $\eval{w^+}_p$ shifted to the left by a
number of zeros specified by an ultimately periodic set.
\begin{lemma}
  Let $w\in \Sigma_p^*$ and $U$ be an ultimately periodic set. Then
  $\eval{w^+0^U}_p$ is definable by a formula of existential B\"uchi
  arithmetic $\Phi_{U,w^+}(x)$.
\end{lemma}
\begin{proof}
  The case $w\in 0^*$ is trivial. Thus, let $w=w'\cdot w_0$ such that
  $w'\in \Sigma_p^*\cdot (\Sigma_p\setminus \{ 0 \})$ and $w_0 \in
  0^*$. Observe that for $i<j$, $\eval{w^j}_p-\eval{w^i}_p =
  \eval{w^{j-i}0^i}_p$. We define
  \begin{multline*}
    \Phi_{U,w^+}(x) := \exists y\,\exists z\colon y<z \wedge
    \Phi_{w^*}(y) \wedge \Phi_{w^*}(z) \wedge \bigvee_{0\le i < |w|} x
    = p^i \cdot (z-y)\wedge\\ \wedge \exists s\, \exists t\colon S_U(1,s)
    \wedge V_p(t,x) \wedge t = p^{|w_0|+1} \cdot s\,.
  \end{multline*}
  The first line defines the set $\eval{w^+0^*}_p$, whereas the second
  line ensures that the tailing number of zeros is in the set
  $U+|w_0|$. \qed
\end{proof}
We have now all the ingredients to prove the following key
proposition.
\begin{proposition}\label{prop:poly-growth-re-ba}
  Let $L=v_0 w_1^* v_{1} \cdots v_{k-1} w_k^* v_k$. Then $\eval L_p$
  is definable in existential B\"uchi arithmetic.
\end{proposition}
\begin{proof}
  The proposition follows from showing the statement for languages of
  the form \[L'=v_0 w_1^+ v_{1} \cdots v_{k-1} w_k^+ v_k\,.\] We show
  the statement by induction on $k$. The induction base case $k=0$ is
  trivial. For the induction step, assume that for $M=v_{1} w_{2}^+
  v_{2} \cdots v_{k-1} w_k^+ v_k$, $\eval{M}_p$ is defined by a
  formula $\Phi_k(x)$ of existential B\"uchi arithmetic, and let
  $v_{0},w_{1} \in \Sigma_p^*$.

  We first show how to define $N=w_1^+v_{1} w_{2}^+ v_{2} \cdots
  v_{k-1} w_k^+ v_k$. To this end, factor $M=M_0 \cdot M'$, where $M_0
  \subseteq 0^*$ and $M\subseteq (\Sigma_p \setminus \{ 0\}) \cdot
  \Sigma_p^*$. Observe that $\eval{M'}_p=\eval{\Phi_k(x)}_p$, and that
  both $U=\{ |w| : w \in M\}$ and $V=\{ |w| : w \in M_0 \}$ are
  ultimately periodic sets, cf.~\cite{Chro86,To09}. We moreover assume
  that $w_1\not\in 0^*$, otherwise we are done. Factor $w_1=w'\cdot
  w_0$ such that $w'\in \Sigma_p^* \cdot (\Sigma_p \setminus \{ 0\})$
  and $w_0\in 0^*$. Recall that $W_p(x,y)$ holds if and only if $y$ is
  the smallest power of $p$ strictly greater than $x$, and define
  \begin{multline*}
    \Psi_{k+1}(x) := \exists y\,\exists z\colon \Phi_k(y) \wedge
    \Phi_{U,w^+}(z) \wedge x = y + z \wedge\\ \wedge \exists
    s\,\exists t\colon W_p(y,s) \wedge S_V(s,t) \wedge V_p(p^{|w_0|+1}
    \cdot t,z)\,.
  \end{multline*}
  The first line composes $x$ as the sum of some $y \in \eval{M}_p$
  and $z\in \eval{w^+0^U}_p$. The second line ensures that the number of
  zeros between the leading bit of $y$ and the last non-zero digit of
  $z$ in their $p$-ary expansion is in $V+|w_0|$. Thus, $\eval
  N_p=\eval{\Psi_{k+1}(x)}$.

  We now show how to define $L'$ along similar lines. To this end,
  factor $N={N_0\cdot N'}$ such that $N_0 \subseteq 0^*$ and
  $N'\subseteq (\Sigma_p \setminus \{0 \})\cdot \Sigma_p^*$, and let
  $T=\{ |w| : w\in N_0\}$, which is an ultimately periodic set.  We
  now obtain the desired formula of existential B\"uchi arithmetic as
  \[
  \Phi_{k+1}(x) := \exists y\, \exists z\colon x = y + p\cdot z \cdot
  \eval{v_0}_p \wedge \Psi_{k+1}(y) \wedge \exists s\colon W_p(y,s) \wedge
  S_T(s,z)\,.
  \]
  \qed
\end{proof}
Since we can define any regular language of the
form~\eqref{eqn:poly-growth-re} in existential B\"uchi arithmetic via
\cref{prop:poly-growth-re-ba}, we can define a finite union of such
languages and thus define all regular languages of polynomial growth
in existential B\"uchi arithmetic. This completes the proof of
\cref{thm:preg-ba}.

Note that $\mathbf{PREG}_p\not\subseteq \mathbf{PA}$ for any base
$p\ge 2$: since $M=\eval{\Phi(x)}$ is ultimately periodic for any
formula $\Phi(x)$ of Presburger arithmetic, whenever $\eval{\Phi(x)}$
is infinite it follows that $d_M(n)=\Omega(p^n)$, i.e., not of
polynomial growth.

\section{Conclusion}
The main result of this paper is that existential B\"uchi arithmetic
is strictly less expressive than full B\"uchi arithmetic of any base.
This is in contrast to Presburger arithmetic, for which it is known
that its existential fragment is expressively complete.

When considered as the first-order theory of the structure $\langle
\mathbb N,0,1,+ \rangle$, Presburger arithmetic does not have a
quantifier elimination procedure. The extended structure $\langle
\mathbb N,0,1,+,\{c|\cdot\}_{c>1} \rangle$, however, admits quantifier
elimination. Those additional divisibility predicates are definable in
existential Presburger arithmetic. Our main result shows that even if
we extended the structure underlying B\"uchi arithmetic with
predicates definable in existential B\"uchi arithmetic, the resulting
first-order theory would not admit quantifier-elimination. On the
positive side, Benedikt et al.~\cite[Thm.~3.1]{BLSS03} give an
extension of B\"uchi arithmetic which has quantifier elimination.

We conclude this paper with an interesting yet likely challenging open
problem: Is it decidable whether a set definable in B\"uchi arithmetic
is definable in existential B\"uchi arithmetic?

\subsubsection*{Acknowledgments.}
We would like to thank Dmitry Chistikov and Alex Fung for inspiring
discussions on the topics of this paper, and the FoSSaCS'21 reviewers
for their comments and suggestions.

\begin{wrapfigure}{r}{20pt}
  \vspace*{-0.88cm}
  \includegraphics[width=20pt]{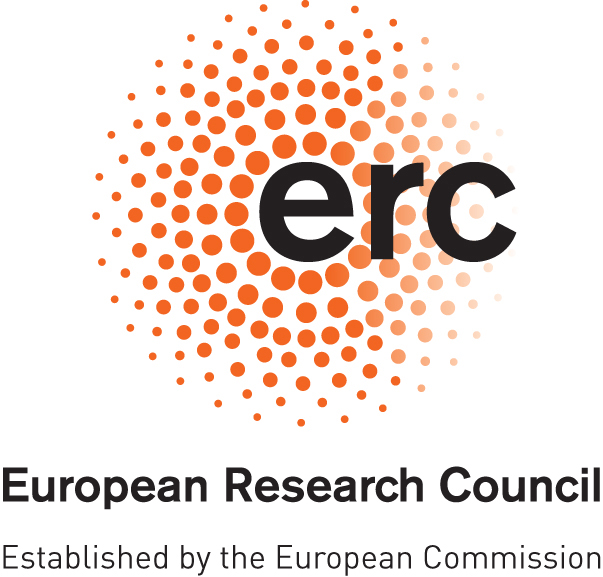}\\
  \includegraphics[width=20pt]{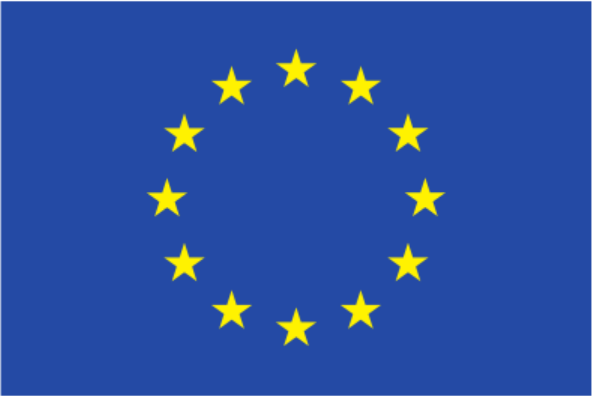}
\end{wrapfigure}

This work is part of a project that has received funding from the
European Research Council (ERC) under the European Union's Horizon
2020 research and innovation programme (Grant agreement No.\ 852769,
ARiAT).

\bibliographystyle{splncs04}
\bibliography{bibliography}

\vfill
 
{\small\medskip\noindent{\bf Open Access} This chapter is licensed under the terms of the Creative Commons\break Attribution 4.0 International License (\url{http://creativecommons.org/licenses/by/4.0/}), which permits use, sharing, adaptation, distribution and reproduction in any medium or format, as long as you give appropriate credit to the original author(s) and the source, provide a link to the Creative Commons license and indicate if changes were made.}
 
{\small \spaceskip .28em plus .1em minus .1em The images or other third party material in this chapter are included in the chapter's Creative Commons license, unless indicated otherwise in a credit line to the material.~If material is not included in the chapter's Creative Commons license and your intended\break use is not permitted by statutory regulation or exceeds the permitted use, you will need to obtain permission directly from the copyright holder.}
 
\medskip\noindent\includegraphics{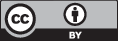}

\end{document}